\documentclass[letterpaper,UKenglish]{lipics-v2016}
 
\usepackage{microtype}
\usepackage{algpseudocode,algorithm} 

\newcommand{\OPT}{\mbox{\sc OPT}}

\def\tn{\textnormal}

\def\mcr{\mathsf{R}}
\def\mbr{\mathbb{R}}
\def\mbfy{\mathbf{y}}

\title{Approximation algorithms for the vertex happiness
\footnote{This work was partially supported by NSERC Canada and NSFC 61672323.}}
\titlerunning{Approximating the vertex happiness} 

\author{Yao Xu}
\author[2]{Peng Zhang}
\author{Randy Goebel}
\author{Guohui Lin}
\affil{Department of Computing Science, University of Alberta.
  Edmonton, Alberta T6G 2E8, Canada.
  \texttt{\{xu2, rgoebel, guohui\}@ualberta.ca}}
\affil[2]{School of Computer Science and Technology, Shandong University.
  Jinan, Shandong 250101, China.
  \texttt{algzhang@sdu.edu.cn}}
\authorrunning{Xu {\it et al.} version/\today} 

\Copyright{Yao Xu, Peng Zhang, Randy Goebel and Guohui Lin} 

\subjclass{Dummy classification -- please refer to \url{http://www.acm.org/about/class/ccs98-html}}

\begin{document}
\maketitle

\begin{abstract}
We investigate the maximum happy vertices (MHV) problem and its complement, the minimum unhappy vertices (MUHV) problem.
We first show that the MHV and MUHV problems are a special case of the supermodular and submodular multi-labeling (Sup-ML and Sub-ML) problems, respectively,
by re-writing the objective functions as set functions.
The convex relaxation on the Lov\'{a}sz extension, originally presented for the submodular multi-partitioning (Sub-MP) problem,
can be extended for the Sub-ML problem,
thereby proving that the Sub-ML (Sup-ML, respectively) can be approximated within a factor of $2 - \frac{2}{k}$ ($\frac{2}{k}$, respectively).
These general results imply that the MHV and the MUHV problems can also be approximated within $\frac{2}{k}$ and $2 - \frac{2}{k}$, respectively,
using the same approximation algorithms.
For MHV, this $\frac{2}{k}$-approximation algorithm improves the previous best approximation ratio $\max \{\frac{1}{k}, \frac{1}{\Delta + 1}\}$, 
where $\Delta$ is the maximum vertex degree of the input graph.
We also show that an existing LP relaxation is the same as the concave relaxation on the Lov\'{a}sz extension for the Sup-ML problem;
we then prove an upper bound of $\frac{2}{k}$ on the integrality gap of the LP relaxation.
These suggest that the $\frac{2}{k}$-approximation algorithm is the best possible based on the LP relaxation.
For MUHV, we formulate a novel LP relaxation and prove that it is the same as the convex relaxation on the Lov\'{a}sz extension for the Sub-ML problem;
we then show a lower bound of $2 - \frac{2}{k}$ on the integrality gap of the LP relaxation.
Similarly, these suggest that the $(2 - \frac{2}{k})$-approximation algorithm is the best possible based on the LP relaxation.
Lastly, we prove that this $(2 - \frac{2}{k})$-approximation is optimal for the MUHV problem, assuming the Unique Games Conjecture.

\end{abstract}

\section{Introduction}
\label{sec:intro}
In a recently studied vertex-coloring problem by Zhang and Li~\cite{ZL15}, 
one is given an undirected graph $G = (V, E)$ with a non-negative weight $w(v)$ for each vertex $v \in V$, a color set $C = \{1, 2, \ldots, k\}$,
and a partial vertex coloring function $c: V \mapsto C$,
and the goal is to color all the uncolored vertices such that the total weight of {\em happy} vertices is maximized.
A vertex is {\em happy} if it shares the same color with all its neighbors in the coloring scheme.
The problem is referred to as the {\em maximum happy vertices} (MHV)~\cite{ZL15},
which was inspired by the study on the {\em homophyly} governing the structures of large scale networks such as social networks and citation networks.

The complement of the MHV problem is the {\em minimum unhappy vertices} (MUHV),
which can be defined analogously to minimize the total weight of unhappy vertices,
where a vertex is {\em unhappy} if its color is different from at least one of its neighbors.

We remark that these two vertex-coloring problems are in fact labeling problems, and we use ``color'' and ``label'' interchangeably in the sequel;
they are different from the classic {\em graph coloring} problem~\cite{Har69},
in which a feasible vertex coloring scheme must assign different colors to any adjacent vertices.
We also note that, if no vertex is pre-colored $i$, for any $i$, then this color $i$ can be removed without affecting the optimum;
we therefore assume without loss of generality that every color is used in the given partial vertex coloring function $c$.

Given the graph $G = (V, E)$ with the vertex set $V$ and the edge set $E$, for any subset $X \subseteq V$,
define the {\em boundary} of $X$, denoted as $\partial(X)$, to be the subset of vertices of $X$ each has at least one neighbor outside of $X$.
Let $\iota(X) = X - \partial(X)$, which is called the {\em interior} of $X$.
In a vertex coloring scheme, let $S_i$ denote the subset of all the vertices colored $i$;
then every vertex of $\partial(S_i)$ is unhappy while all vertices of $\iota(S_i)$ are happy.
We extend the vertex weight function to subsets of vertices, that is, $w(X) := \sum_{v \in X} w(v)$ for any $X \subseteq V$;
and we define the set function $f(\cdot)$ as
\begin{equation}
\label{eq1}
f(X) := w(\partial(X)), \ \forall X \subseteq V.
\end{equation}
Note that a vertex coloring scheme one-to-one corresponds to a partition ${\cal S} = \{S_1, S_2, \ldots, S_k\}$ of the vertex set $V$,
where each part $S_i$ contains all the vertices colored $i$.
This way, the MUHV problem can be cast as finding a partition ${\cal S}$ such that $f({\cal S}) := \sum_{i=1}^k f(S_i)$ is minimized.

It is important to note that the above defined set function $f(\cdot)$ depends on the given edge set $E$;
a change to $E$ could alter the function definition, and subsequently alters the optimization objective.
In particular, when there are multiple vertices in the graph pre-colored the same color,
we cannot simply contract all of them into a single vertex unless they have exactly the same neighbors in the original graph;
otherwise, this contracting process essentially changes the edge set $E$, causing a change to the defined set function $f(\cdot)$.
(A concrete example is provided in the Appendix C.)

It is not hard to validate (the proofs are provided in the Appendix A)
that the boundary $\partial(\cdot)$ of a vertex subset in the given graph $G = (V, E)$ has the following properties:
i) $\partial(\emptyset) = \emptyset$;
ii) $\partial(X \cap Y) \subseteq \partial(X) \cup \partial(Y)$;
iii) $\partial(X \cup Y) \subseteq \partial(X) \cup \partial(Y)$; and
iv) $\partial(X \cap Y) \cap \partial(X \cup Y) \subseteq \partial(X) \cap \partial(Y)$, for any two subsets $X, Y \subseteq V$.

Therefore, the set function $f: 2^V \to \mbr$ defined in Eq.~(\ref{eq1}) satisfies $f(X) + f(Y) \ge f(X \cap Y) + f(X \cup Y)$,
for any two subsets $X, Y \subseteq V$ (a detailed proof is provided in the Appendix B).
That is, $f(\cdot)$ is a {\em submodular}~\cite{L83} function on the set $V$.
This way, the MUHV problem can be cast as a special case of the following submodular multi-labeling (Sub-ML) problem:

Given a ground set $V$, a non-negative submodular set function $f: 2^V \to \mbr$ with $f(\emptyset) = 0$,
a set of labels $L = \{1, 2, \ldots, k\}$,
and a partial labeling function $\ell: V \mapsto L$ which pre-assigns each label $i$ to a non-empty subset $T_i \subset V$,
the goal of the {\em submodular multi-labeling} (Sub-ML) problem is to find a partition ${\cal S} = \{S_1, S_2, \ldots, S_k\}$ of the ground set $V$
to minimize $f({\cal S}) := \sum_{i = 1}^{k} f(S_i)$,
where the part $S_i$ is the subset of elements assigned with the label $i$.

We remark again that for each $i$, $|T_i| \ge 1$,
and any attempt to contracting $T_i$ into a single element could either destroy the submodularity of the set function $f(\cdot)$
or alter the function definition leading to a change in the optimization objective.

Conversely, given the graph $G = (V, E)$
we define another set function $g(\cdot)$ as
\begin{equation}
\label{eq2}
g(X) := w(\iota(X)), \ \forall X \subseteq V.
\end{equation}
Then $g(X) = w(X) - f(X)$ for any subset $X \subseteq V$, and consequently $g(\cdot)$ is a {\em supermodular}~\cite{L83} function on the set $V$.
Thus, the MHV problem can be cast as finding a partition ${\cal S} = \{S_1, S_2, \ldots, S_k\}$ of the vertex set $V$
such that $g({\cal S}) := \sum_{i=1}^k g(S_i)$ is maximized,
where each part $S_i$ contains all the vertices colored $i$;
it can also be cast as a special case of the {\em supermodular multi-labeling} (Sup-ML) problem that can be analogously defined.

\subsection{Related research}
Classification problems have been formulated as cuts, or partition, or labeling, or coloring, and have been widely studied for a very long time.

For the MHV problem, Zhang and Li~\cite{ZL15} proved that it is polynomial time solvable for $k = 2$ and it becomes NP-hard for $k \ge 3$;
for $k \ge 3$, they presented two approximation algorithms:
a greedy algorithm with an approximation ratio of $\frac{1}{k}$,
and an $\Omega(\frac{1}{\Delta^3})$-approximation based on a subset-growth technique, where $\Delta$ is the maximum vertex degree of the input graph.
Recently, Zhang {\it et al.}~\cite{ZJL15} presented an improved algorithm with an approximation ratio of $\frac{1}{\Delta + 1}$ based on 
a combination of randomized LP rounding techniques.
Together, these imply that the current best approximation ratio for the MHV problem is $\max \{\frac{1}{k}, \frac{1}{\Delta + 1}\}$.

For the complementary MUHV problem, to the best of our knowledge, it hasn't been particularly studied in the literature.

Recall that the MHV and the MUHV problems are a special case of the Sup-ML and the Sub-ML problems, respectively.
We again remind the readers that in an instance of these multi-labeling problems,
each label is pre-assigned to at least one element and to multiple elements in general.
Another special case of the Sub-ML problem is when each label is pre-assigned to exactly one element,
called the {\em submodular multiway partition} (Sub-MP) problem~\cite{ZNI05}, which has received a lot of studies.
In the Appendix C, we provide an instance to show that one cannot reduce the MUHV problem to the Sub-MP problem
by simply contracting all elements pre-assigned the same label into a single element.

The Sub-MP problem was first studied by Zhao {\it et al.} \cite{ZNI05}, who presented a $(k - 1)$-approximation algorithm.
Years later, Chekuri and Ene \cite{CE11smp} proposed a convex relaxation for Sub-MP by using the Lov\'{a}sz extension, leading to a 2-approximation.
This was further improved to a $(2 - \frac{2}{k})$-approximation shortly after by Ene {\it et al.}~\cite{EVW13}.
On the inapproximability, Ene {\it et al.}~\cite{EVW13} proved that any $(2 - \frac{2}{k} - \epsilon)$-approximation for Sub-MP
requires exponentially many value queries, for any $\epsilon > 0$, or otherwise it implies $NP = RP$.

Sub-MP includes many well studied cut problems including the classic (edge-weighted) multiway cut~\cite{DJPSY94},
the node-weighted multiway cut~\cite{GVY04} and the hypergraph multiway cut~\cite{OFN12} as special cases.
The multiway cut problem is NP-hard for $k \ge 3$ even if all edges have unit weight~\cite{DJPSY94},
with many approximation algorithms designed and analyzed~\cite{DJPSY94,CKR98,FK00,KKSTY04,BNS13,SV14}.
Most of these approximation results are based on the {\em linear program} (LP) relaxation presented by C{\u{a}}linescu {\it et al.}~\cite{CKR98},
and the current best approximation ratio is $1.2965$~\cite{SV14}.
The hypergraph multiway cut and the node-weighted multiway cut are proven more difficult to approximate,
that it is Unique Games-hard to achieve a $(2 - \frac{2}{k} - \epsilon)$-approximation for any $\epsilon > 0$~\cite{EVW13}.

One can similarly define the complement of the Sub-MP problem, called the {\em supermodular multiway partition} (Sup-MP) problem.
Sup-MP includes the multiway uncut problem~\cite{LRS06} as a special case,
where the $k$ terminals in the input graph can be considered as $k$ elements each being pre-assigned with a distinct label.
The multiway uncut problem seems only studied by Langberg {\it et al.}~\cite{LRS06}, who presented a $0.8535$-approximation based on an LP relaxation.
When generalizing the multiway uncut problem to pre-assign multiple terminals in a part of the vertex partition,
it becomes the recently studied {\em maximum happy edges} (MHE) problem~\cite{ZL15}.
It is important to note that MHE is not a special case of the Sup-MP problem, but a special case of the Sup-ML problem.
Zhang and Li~\cite{ZL15} presented a $\frac 12$-approximation for the MHE problem based on a simple division strategy;
extending the LP relaxation for the multiway uncut,
Zhang {\it et al.}~\cite{ZJL15} improved the approximation ratio to $\frac 12 + \frac{\sqrt{2}}{4} h(k) \ge 0.8535$,
where $h(k) \ge 1$ is a function in $k$.

More broadly, the multi-labeling problems can be deemed as special cases of the {\em cost allocation} (CA) problem~\cite{CE11sca},
in which $k$ different non-negative set functions are given for evaluating the $k$ parts of the partition separately;
they are also closely related to the {\em optimal allocation} (OA) problem~\cite{LLN01,LOS02,FV06,DS06,F09,KLMM08,SS12} in combinatorial auctions,
where no elements are necessarily pre-assigned a label but the set function (called {\em utility function}) is assumed monotone in general.

\subsection{Our contributions}
Our target problems are the MHV and the MUHV problems,
and we aim to design improved approximation algorithms for them and to prove the hardness results in approximability.

We first show that the convex relaxation on the Lov\'{a}sz extension for the Sub-MP problem~\cite{CE11smp} can be extended for the Sub-ML problem;
therefore the same approximation algorithm works for Sub-ML with a performance ratio $(2 - \frac{2}{k})$.
Analogously, we present the concave relaxation on the Lov\'{a}sz extension for the Sup-ML problem,
thus proving that Sup-ML can be approximated within a factor of $\frac{2}{k}$.
Therefore, the MUHV problem can be approximated within a factor of $(2 - \frac{2}{k})$ and the MHV problem can be approximated within a factor of $\frac 2k$ too;
the $\frac 2k$-approximation algorithm for the MHV problem improves the previous best ratio of $\max\{\frac 1k, \frac 1{\Delta + 1}\}$~\cite{ZL15,ZJL15}.

Next, for the MHV problem, we show that the LP relaxation presented in \cite{ZJL15}, called LP-MHV,
is equivalent to the concave relaxation for the Sup-ML problem based on the Lov\'{a}sz extension to the set function $g(\cdot)$ defined in Eq.~(\ref{eq2});
for the MUHV problem, we propose a novel LP relaxation, called LP-MUHV,
and we show that it is equivalent to the convex relaxation for Sub-ML based on the Lov\'{a}sz extension to the set function $f(\cdot)$ defined in Eq.~(\ref{eq1}).
We then prove an upper bound of $\frac{2}{k}$ on the integrality gap of LP-MHV;
and conclude that the $\frac{2}{k}$-approximation is the best possible based on LP-MHV.
We also prove a lower bound of $2 - \frac{2}{k}$ on the integrality gap of LP-MUHV;
and conclude that the $(2 - \frac{2}{k})$-approximation is the best possible based on LP-MUHV.
Lastly, we prove that it is Unique Games-hard to achieve a $(2 - \frac{2}{k} - \epsilon)$-approximation for MUHV, for any $\epsilon > 0$.
We remark that the last hardness result gives another evidence that
it is Unique Games-hard to achieve a $(2 - \frac{2}{k} - \epsilon)$-approximation for the general Sub-ML problem, for any $\epsilon > 0$.

\subsection{Organization}
The remainder of the paper is organized as follows.
In the next section, we introduce some basic notions such as the Lov\'{a}sz extension to a set function;
we then present the relaxation based on the Lov\'{a}sz extension for the Sub-ML problem
and a similar relaxation for the Sup-ML problem.
We also present a simple approximation algorithm using the randomized rounding technique in \cite{EVW13},
and conclude that it is a $(2 - \frac{2}{k})$-approximation for the Sub-ML problem and it is a $\frac 2k$-approximation for the Sup-ML problem.
In Section 3, we study the MHV problem, by firstly introducing the LP relaxation formulated in \cite{ZJL15},
then showing its equivalence to the relaxation based on the Lov\'{a}sz extension to the set function $f(\cdot)$ defined in Eq.~(\ref{eq1}),
and lastly proving an upper bound of $\frac 2k$ on the integrality gap.
In Section 4, we first present a novel LP relaxation for the MUHV problem,
then show its equivalence to the relaxation based on the Lov\'{a}sz extension to the set function $g(\cdot)$ defined in Eq.~(\ref{eq2}),
then similarly prove a lower bound of $(2 - \frac 2k)$ on the integrality gap,
and lastly prove a stronger inapproximability result that it is Unique Games-hard to achieve a $(2 - \frac 2k - \epsilon)$-approximation, for any $\epsilon > 0$.
We conclude the paper in Section 5.

\section{Preliminaries}
Given a ground set $V = \{v_1, v_2, \ldots, v_n\}$, $y_j := y(v_j)$ is a real variable that maps the element $v_j$ to the closed unit interval $[0, 1]$.
For any non-negative set function $f: 2^V \to \mbr$, its Lov\'{a}sz extension~\cite{L83,V13} is a function $\hat{f}: [0, 1]^V \to \mbr$ such that
\begin{equation}
\label{eq3}
\hat{f}({\mbfy}) := \sum_{j = 1}^{n-1} (y_{\pi_j} - y_{\pi_{j+1}}) f(\{v_{\pi_1}, v_{\pi_2}, \ldots, v_{\pi_j}\}),
\end{equation}
where $\mbfy = (y_1, y_2, \ldots, y_n) \in [0, 1]^V$ and
$\pi$ is a permutation on $\{1, 2, \ldots, n\}$ such that $1 = y_{\pi_1} \ge y_{\pi_2} \ge \ldots \ge y_{\pi_n} = 0$.

It has been proven by Lov{\'a}sz~\cite{L83} that the set function $f(\cdot)$ is submodular (supermodular, respectively)
if and only if its Lov\'{a}sz extension is convex (concave, respectively).

In the context of the Sub-ML problem with $f(\cdot)$ being the non-negative submodular set function
and $T_i \subset V$ being the non-empty subset of elements pre-labeled $i$, $i \in L = \{1, 2, \ldots, k\}$,
we define a binary variable $y_j^i := y^i(v_j)$ for each pair of an element $v_j$ and a label $i$,
such that $y_j^i = 1$ if and only if the element $v_j$ is labeled $i$.
Next, $y_j^i$ is relaxed to be a real variable in the closed unit interval $[0, 1]$.
For each $i$, let $\mbfy_i = (y_1^i, y_2^i, \ldots, y_n^i) \in [0, 1]^V$;
let $\hat{f}: [0, 1]^V \to \mbr$ be the Lov\'{a}sz extension of $f(\cdot)$ as defined in Eq.~(\ref{eq3}).

A relaxation based on the Lov\'{a}sz extension for the Sub-ML problem can be written as follows:
\begin{alignat}{4}
& \tn{minimize} \quad 	& \sum_{i = 1}^k \hat{f}(\mbfy_i) \tag{\tn{CP-Sub-ML}} \label{cp:subml}\\
& \tn{subject to} \quad & \sum_{i = 1}^k y_j^i 	& = 1, \quad 	& \forall v_j 	& \in V \label{cp:subml:1}\\
&						& y_j^i 				& = 1, \quad 	& \forall v_j	& \in T_i, ~i \in L \label{cp:subml:t}\\
&						& y_j^i 				& \ge 0, \quad 	& \forall v_j 	& \in V, ~i \in L
\end{alignat}
The submodularity of the function $f(\cdot)$ implies that (\ref{cp:subml}) is a {\em convex program} (CP) and thus can be solved in polynomial time.

In fact, such a relaxation based on the Lov\'{a}sz extension was proposed by Chekuri and Ene~\cite{CE11smp} for the Sub-MP problem,
which is a special case of the Sub-ML problem in that $|T_i| = 1$ for every label $i$.
We extend this relaxation for the Sub-ML problem with little change, except that in the Constraint~(\ref{cp:subml:t}) $y_j^i = 1$ holds for multiple elements $v_j$.
Nevertheless, we remark again that one cannot reduce the Sub-ML problem to Sub-MP by
cruelly contracting all the elements pre-labeled with the same label into a single element,
which suggests incorrectly that all these pre-labeled elements were identical.

The following approximation algorithm {\bf $\mcr\mcr$} first solves the convex program (\ref{cp:subml}),
followed by a randomized rounding scheme to obtain a feasible solution to the Sub-ML problem.
Ene {\it et al.} showed that {\bf $\mcr\mcr$} is a $(2 - \frac 2k)$-approximation algorithm for the Sub-MP problem~\cite{EVW13}.
The algorithm uses a uniformly random variable $\theta$ in the interval $(\frac 12, 1]$, and defines the following $k+3$ sets:
\begin{equation}
\label{eq:def}
\begin{array}{lcl}
S_i(\theta)	&= &\{v_j \mid y_j^i > \theta \}, \ \mbox{for each } i \in L,\\
S(\theta) 	&= &\bigcup_{i=1}^k S_i(\theta),\\
R(\theta) 	&= &V - S(\theta),\\
Q(\theta) 	&= &R(1 - \theta).
\end{array}
\end{equation}

\begin{algorithm}[htb]
\caption*{{\bf Algorithm $\mcr\mcr$}}
\begin{algorithmic}[1]
\State Solve (\ref{cp:subml}) to obtain an optimal fractional solution $\{y_j^i \mid v_j \in V, i \in L\}$. \label{alg:rr:lp}
\State Pick a parameter $\theta \in (\frac{1}{2}, 1]$ uniformly at random.
\State Assign all elements of $S_i(\theta)$ the label $i$, for each $i \in L$.
\State Pick a label $i'$ from $L$ uniformly at random, assign all elements of $R(\theta)$ the label $i'$.
\end{algorithmic}
\end{algorithm}

The performance analysis for the algorithm {\bf $\mcr\mcr$} on the Sub-MP problem in \cite{EVW13} does not need the fact that $|T_i| = 1$ for every label $i$.
Therefore, the same analysis proves the following theorem.

\begin{theorem}{\rm \cite{EVW13}}
\label{thm1}
Algorithm $\mcr\mcr$ is a $\left(2 - \frac{2}{k}\right)$-approximation for the Sub-ML problem.
\end{theorem}

Replacing the submodular function $f(\cdot)$ by a supermodular function $g(\cdot)$ and inverting the minimization to the maximization,
a relaxation based on the Lov\'{a}sz extension for Sup-ML can be written as follows:
\begin{alignat}{4}
& \tn{maximize} \quad 	& \sum_{i = 1}^k \hat{g}(\mbfy_i) \tag{\tn{CP-Sup-ML}} \label{cp:supml}\\
& \tn{subject to} \quad & \sum_{i = 1}^k y_j^i 	& = 1, \quad 	& \forall v_j 	& \in V \label{cp:supml:1}\\
&						& y_j^i 				& = 1, \quad 	& \forall v_j	& \in T_i, ~i \in L \label{cp:supml:t}\\
&						& y_j^i 				& \ge 0, \quad 	& \forall v_j 	& \in V, ~i \in L
\end{alignat}
where $\hat{g}: [0, 1]^V \to \mbr_+$ is the Lov\'{a}sz extension of $g(\cdot)$ as defined in Eq.~(\ref{eq3}).
(\ref{cp:supml}) is a {\em concave program} and thus can be solved in polynomial time.
Using an analogous argument as the proof of Theorem~\ref{thm1}, we can have the following corollary on the Sup-ML problem.

\begin{corollary}
\label{coro2}
Algorithm $\mcr\mcr$ is a $\frac 2k$-approximation for the Sup-ML problem.
\end{corollary}

\section{The maximum happy vertices (MHV) problem}
\label{sec:mhv}
Recall that the MHV problem can be cast as finding a partition ${\cal S} = \{S_1, S_2, \ldots, S_k\}$ of the vertex set $V$
such that $g({\cal S}) = \sum_{i=1}^k g(S_i)$ is minimized,
where the set function $g(\cdot)$ is defined in Eq.~(\ref{eq2}) and $S_i$ is the subset of vertices colored $i$, for each $i$.

\begin{lemma}
\label{lemma3}
The set function $g(\cdot)$ defined in Eq.~(\ref{eq2}) is supermodular.
\end{lemma}

\begin{theorem}
\label{thm4}
Algorithm $\mcr\mcr$ is a $\frac{2}{k}$-approximation for the MHV problem, which is a special case of the Sup-ML problem.
\end{theorem}

The following LP relaxation for the MHV problem (\ref{lp:mhv}), given a graph $G = (V, E)$, is formulated by Zhang {\it et al.}~\cite{ZJL15},
where $V = \{v_1, v_2, \ldots, v_n\}$,
$w_j = w(v_j)$ denotes the weight of the vertex $v_j$,
$C = \{1, 2, \ldots, k\}$ is the color set,
$c(v_j) = i$ if the vertex $v_j$ is pre-colored $i$,
a binary variable $y_j^i := y^i(v_j)$ denotes whether or not the vertex $v_j$ is colored $i$,
and $\mbfy_i = (y_1^i, y_2^i, \ldots, y_n^i)$.

\begin{alignat}{5}
& \tn{maximize} \quad & \sum_{j=1}^n{w_j z_j} \tag{\tn{LP-MHV}} \label{lp:mhv}\\
& \tn{subject to} \quad & \sum_{i=1}^k {y_j^i} 			& = 1, \quad & \forall v_j & \in V\\
&	& y_j^i & = 1, \quad & \forall v_j 					& \in V, ~\forall i \in C ~\tn{s.t.} ~c(v_j) = i\\
&	& z_j^i & = \min_{v_h \in N[v_j]} \{y_h^i\}, \quad 	& \forall v_j & \in V, ~\forall i \in C\label{eq32}\\
&	& z_j 	& = \sum_{i=1}^k z_j^i, \quad 				& \forall v_j & \in V\label{lp:mhv:z}\\
&	& z_j, ~z_j^i, ~y_j^i & \ge 0, \quad 				& \forall v_j & \in V, ~\forall i \in C
\end{alignat}
where $z_j^i$ indicates whether the vertex $v_j$ is happy by color $i$, $z_j$ indicates whether the vertex $v_j$ is happy,
and $N[v_j]$ is the closed neighborhood of the vertex $v_j$.

For each color $i$, since there is at least one vertex pre-colored $i$ and at least one vertex pre-colored another color (due to $k \ge 2$),
we let $\pi$ be the permutation for $\mbfy_i$ such that $1 = y_{\pi_1}^i \ge y_{\pi_2}^i \ge \ldots \ge y_{\pi_n}^i = 0$.
In the concave relaxation (\ref{cp:supml}) based on the Lov\'{a}sz extension for Sup-ML,
when we set the supermodular set function $g(\cdot)$ as in Eq.~(\ref{eq2}), the objective function of (\ref{cp:supml}) becomes
\begin{align}
\sum_{i = 1}^k \hat{g}(\mbfy_i)
	& = \sum_{i = 1}^k \sum_{j=1}^{n-1} \left(y^i_{\pi_j} - y^i_{\pi_{j+1}}\right) g(\{v_{\pi_1}, v_{\pi_2}, \ldots, v_{\pi_j}\}) \notag\\
	& = \sum_{i = 1}^k \sum_{j=1}^{n-1} \left(y^i_{\pi_j} - y^i_{\pi_{j+1}}\right) \sum_{v_h \in \iota(\{v_{\pi_1}, v_{\pi_2}, \ldots, v_{\pi_j}\})} w_h.
	\label{eq:muhv:cg}
\end{align}

For each vertex $v_p \in V$, let $v_q$ denote its neighbor that appears the last in the permutation $(v_{\pi_1}, v_{\pi_2}, \ldots, v_{\pi_n})$.
Assume $p = \pi_{j_1}$ and $q = \pi_{j_2}$.
Clearly, $v_p \in \iota(\{v_{\pi_1}, v_{\pi_2}, \ldots, v_{\pi_j}\})$ if and only if $p, q \in \{\pi_1, \pi_2, \ldots, \pi_j\}$,
that is, we must have $j_1, j_2 \le j$.
It follows that for the vertex $v_p \in V$, the coefficient of $w_p$ in Eq.~(\ref{eq:muhv:cg}) is
\begin{equation*}
	\sum_{i = 1}^k \sum_{j = \max\{j_1, j_2\}}^n \left(y^i_{\pi_j} - y^i_{\pi_{j+1}}\right)
= 	\sum_{i = 1}^k z^i_p
=	z_p,
\end{equation*}
where the last two equalities hold due to Constraints~(\ref{eq32}, \ref{lp:mhv:z}) of (\ref{lp:mhv}).
This shows that by setting the supermodular set function $g(\cdot)$ as defined in Eq.~(\ref{eq2}), (\ref{cp:supml}) is the same as (\ref{lp:mhv}).
Therefore, we have the following theorem.

\begin{theorem}
\label{thm5}
The LP relaxation for the MHV problem (\ref{lp:mhv}) is the same as the relaxation based on the Lov\'{a}sz extension for the Sup-ML problem (\ref{cp:supml}),
when the MHV problem is cast into the Sup-ML problem.
\end{theorem}

We construct an instance $I = (G = (V, E), w(\cdot), C = \{1, 2, \ldots, k\}, c)$ of the MHV problem
to obtain an upper bound on the integrality gap of (\ref{lp:mhv}), the LP relaxation for the MHV problem. 

\begin{itemize}
\item
	Let $T = \{t_1, t_2, \ldots, t_k\}$ be a set of $k$ pre-colored vertices, called {\em terminals};
	all terminals have the same weight $w_t \ge 0$, and the terminal $t_i$ is pre-colored $i$, {\it i.e.} $c(t_i) = i$.
\item
	Associated with each pair of distinct terminals $t_i$ and $t_j$, $i < j$, there is a vertex $b_{\{ij\}}$.
	Let $V_b = \{b_{\{ij\}} \mid i < j\}$, then $|V_b| = {k \choose 2}$;
	all vertices of $V_b$ have the same weight $w_b \ge 0$, and none of them is pre-colored.
\item
	The vertex set $V = T \cup V_b$;
	the edge set $E = \{\{t_i, b_{\{ij\}}\}, \{t_j, b_{\{ij\}}\} \mid i < j\}$.
	Clearly, $|V| = k + {k \choose 2}$ and $|E| = 2 {k \choose 2}$.
\end{itemize}

Let $c^*$ denote a coloring function that completes the given partial coloring function $c$,
that is, $c^*$ assigns a color for each vertex and it assigns the color $i$ to the terminal $t_i$, for each $i \in C$.
Then,
\begin{itemize}
\item
	all vertices of $V_b$ must be unhappy, since the vertex $b_{\{ij\}}$ is adjacent to two terminals $t_i$ and $t_j$ colored with distinct colors;
\item
	the terminal $t_i$ is adjacent to $k - 1$ vertices $\{b_{\{ij\}} \mid j \ne i\}$, while the vertex $b_{\{ij\}}$ is adjacent to the terminals $t_i$ and $t_j$;
	it follows that if $t_i$ is happy, then all vertices of $\{b_{\{ij\}} \mid j \ne i\}$ are colored $i$,
	subsequently none of the other terminals can be happy;
	in other words, at most one of the $k$ terminals can be happy, regardless of what the coloring function $c^*$ is.
\end{itemize}

Let $\OPT(\tn{MHV})$ denote the value of an optimal solution to the constructed instance $I$;
we obtain
\begin{equation}
\label{eq36}
\OPT(\tn{MHV}) \le w_t.
\end{equation}

Consider the following fractional feasible solution to the instance $I$ in the LP relaxation (\ref{lp:mhv}):
\begin{itemize}
\item
	for each terminal $t_i \in T$, $y^i(t_i) = 1$ and $y^j(t_i) = 0$ for all $j \ne i$;
\item
	for each vertex $b_{\{ij\}} \in V_b$, $y^i(b_{\{ij\}}) = y^j(b_{\{ij\}}) = \frac 12$ and $y^\ell(b_{\{ij\}}) = 0$ for all $\ell \ne i, j$;
\item
	for each terminal $t_i \in T$, we set $z^i(t_i) = y^i(b_{\{ij\}}) = \frac 12$, 
	$z^j(t_i) = 0$ for all $j \ne i$, 
	and $z(t_i) = \sum_{\ell=1}^k z^\ell(t_i) = \frac 12$;
\item
	for each vertex $b_{\{ij\}} \in V_b$, we set $z^\ell(b_{\{ij\}}) = 0$ for all $\ell \in C$, and $z(b_{\{ij\}}) = 0$.
\end{itemize}

Let $\OPT(\tn{LP-MHV})$ denote the optimum of the instance $I$ in the LP relaxation (\ref{lp:mhv}).
It is greater than or equal to the value of the above fractional feasible solution, that is,
\begin{equation}
\label{eq37}
\OPT(\tn{LP-MHV}) \ge \frac 12 k w_t.
\end{equation}

Combining Eqs.~(\ref{eq36}) and (\ref{eq37}), it gives an upper bound on the integrality gap of (\ref{lp:mhv}):
\[
\frac {\OPT(\tn{MHV})}{\OPT(\tn{LP-MHV})} \le \frac {1}{\frac 12 k} = \frac 2k.
\]
We thus have proved the following theorem.

\begin{theorem}
\label{thm6}
The integrality gap of (\ref{lp:mhv}) has an upper bound of $\frac 2k$.
\end{theorem}

Theorems \ref{thm4} and \ref{thm6} together imply that the $\frac 2k$-approximation algorithm $\mcr\mcr$ for the MHV problem is the best possible
based on the LP relaxation (\ref{lp:mhv}), and furthermore

\begin{corollary}
\label{coro7}
The $\frac 2k$-approximation algorithm $\mcr\mcr$ for the Sup-ML problem is the best possible
based on the concave relaxation on the Lov\'{a}sz extension (\ref{cp:supml}).
\end{corollary}

\section{The minimum unhappy vertices (MUHV) problem}
\label{sec:muhv}
Recall that the MUHV problem can be cast as finding a partition ${\cal S} = \{S_1, S_2, \ldots, S_k\}$ of the vertex set $V$
such that $f({\cal S}) = \sum_{i=1}^k f(S_i)$ is minimized,
where the set function $f(\cdot)$ is defined in Eq.~(\ref{eq1}) and $S_i$ is the subset of vertices colored $i$, for each $i$.

\begin{lemma}
\label{lemma8}
The set function $f(\cdot)$ defined in Eq.~(\ref{eq1}) is submodular.
\end{lemma}

\begin{theorem}
\label{thm9}
Algorithm $\mcr\mcr$ is a $(2 - \frac{2}{k})$-approximation for the MUHV problem, which is a special case of the Sub-ML problem.
\end{theorem}

Given an instance of the MUHV problem,
we use a binary variable $y_j^i := y^i(v_j)$ to denote whether or not the vertex $v_j$ is colored $i$,
and $\mbfy_i = (y_1^i, y_2^i, \ldots, y_n^i)$.
We can then formulate a novel LP relaxation as follows.
\begin{alignat}{6}
& \tn{minimize} \quad 	& \sum_{j=1}^n w_j x_j \tag{\tn{LP-MUHV}}\label{lp:muhv}\\
& \tn{subject to} \quad & \sum_{i=1}^k y_j^i 	&= 1, 						&\forall v_j & \in V\\
&						& y_j^i 				&= 1, 						&\forall v_j & \in V, ~\forall i \in C ~\tn{s.t.} ~c(v_j) = i\\
&						& x_j^i 				&\ge y_j^i - y_h^i, 		&\quad\forall v_j & \in V, ~\forall v_h \in N(v_j), ~\forall i \in C \label{eq14}\\
&						& x_j 					&= \sum_{i = 1}^k x_j^i, 	&\forall v_j & \in V \label{eq15}\\
&						& y_j^i, ~x_j^i, ~x_j 	&\ge 0,						&\forall v_j & \in V, ~\forall i \in C
\end{alignat}
where $x_j$ indicates whether the vertex $v_j$ is unhappy, and $N(v_j)$ is the set of neighbors of $v_j$.

For each color $i$, noting there is at least one vertex pre-colored $i$ and at least one vertex pre-colored another color (due to $k \ge 2$),
we let $\pi$ be the permutation on $\{1, 2, \ldots, n\}$ for $\mbfy_i$ such that $1 = y_{\pi_1}^i \ge y_{\pi_2}^i \ge \ldots \ge y_{\pi_n}^i = 0$.
Then by setting the submodular set function $f(\cdot)$ as defined in Eq.~(\ref{eq1}), based on the definition of the Lov\'{a}sz extension in Eq.~(\ref{eq3}),
the objective function of the relaxation (\ref{cp:subml}) becomes
\begin{align}
\sum_{i = 1}^k \hat{f}(\mbfy_i)
	& = \sum_{i=1}^k \left(\sum_{j=1}^{n-1} \left(y^i_{\pi_j} - y^i_{\pi_{j+1}}\right) f(\{v_{\pi_1}, v_{\pi_2}, \ldots, v_{\pi_j}\})\right) \notag\\
	& = \sum_{i=1}^k \sum_{j=1}^{n-1} \left(y^i_{\pi_j} - y^i_{\pi_{j+1}}\right) \sum_{v_h \in \partial(\{v_{\pi_1}, v_{\pi_2}, \ldots, v_{\pi_j}\})} w_h. \label{eq:muhv:cf}
\end{align}
For each vertex $v_p \in V$, let $v_q$ denote its neighbor that appears the last in the permutation $(v_{\pi_1}, v_{\pi_2}, \ldots, v_{\pi_n})$.
Assume $p = \pi_{j_1}$ and $q = \pi_{j_2}$.
Clearly, $v_p \in \partial(\{v_{\pi_1}, v_{\pi_2}, \ldots, v_{\pi_j}\})$ if and only if
i) $p \in \{\pi_1, \pi_2, \ldots, \pi_j\}$ and 
ii) $q \notin \{\pi_1, \pi_2, \ldots, \pi_j\}$,
that is, we must have $j_1 \le j < j_2$.
It follows that for the vertex $v_p \in V$, the coefficient of $w_p$ in Eq.~(\ref{eq:muhv:cf}) is
\begin{equation*}
	\sum_{i = 1}^k \sum_{j = j_1}^{j_2 - 1} \left(y^i_{\pi_j} - y^i_{\pi_{j+1}}\right)
= 	\sum_{i = 1}^k \left(y^i_p - y^i_q\right)
= 	\sum_{i = 1}^k x^i_p
=	x_p,
\end{equation*}
where the last two equalities hold due to Constraints~(\ref{eq14}, \ref{eq15}) of (\ref{lp:muhv}).
This shows that by setting the submodular set function $f(\cdot)$ as defined in Eq.~(\ref{eq1}), (\ref{cp:subml}) is the same as (\ref{lp:muhv}).
Therefore, we have the following theorem.

\begin{theorem}
\label{thm10}
The LP relaxation for the MUHV problem (\ref{lp:muhv}) is the same as the relaxation based on the Lov\'{a}sz extension for the Sub-ML problem (\ref{cp:subml}),
when the MUHV problem is cast into the Sub-ML problem.
\end{theorem}

We use the same instance $I = (G = (V, E), w(\cdot), C = \{1, 2, \ldots, k\}, c)$ constructed in the last section
to obtain a lower bound on the integrality gap of (\ref{lp:muhv}), the LP relaxation for the MUHV problem. 
Let $\OPT(\tn{MUHV})$ denote the optimum of the above constructed instance $I$.
From Eq.~(\ref{eq36}) we have
\begin{equation}
\label{eq18}
\OPT(\tn{MUHV}) \ge (k - 1) w_t + {k \choose 2} w_b.
\end{equation}

Let us consider the following fractional feasible solution to the instance $I$ in the LP relaxation (\ref{lp:muhv}):
\begin{itemize}
\item
	for each terminal $t_i \in T$, $y^i(t_i) = 1$ and $y^j(t_i) = 0$ for all $j \ne i$;
\item
	for each vertex $b_{\{ij\}} \in V_b$, $y^i(b_{\{ij\}}) = y^j(b_{\{ij\}}) = \frac 12$ and $y^\ell(b_{\{ij\}}) = 0$ for all $\ell \ne i, j$;
\item
	for each terminal $t_i \in T$, we set $x^i(t_i) = y^i(t_i) - y^i(b_{\{ij\}}) = \frac 12$, 
	$x^j(t_i) = 0$ for all $j \ne i$, 
	and $x(t_i) = \sum_{\ell=1}^k x^\ell(t_i) = \frac 12$;
\item
	for each vertex $b_{\{ij\}} \in V_b$, we set $x^i(b_{\{ij\}}) = y^i(b_{\{ij\}}) - y^i(t_j) = \frac 12$, 
	$x^j(b_{\{ij\}}) = y^j(b_{\{ij\}}) - y^j(t_i) = \frac 12$, 
	$x^\ell(b_{\{ij\}}) = 0$ for all $\ell \ne i, j$,
	and $x(b_{\{ij\}}) = \sum_{\ell=1}^k x^\ell(b_{\{ij\}}) = 1$.
\end{itemize}

Let $\OPT(\tn{LP-MUHV})$ denote the optimum of the instance $I$ in the LP relaxation (\ref{lp:muhv}).
It is no greater than the value of the above fractional feasible solution, that is,
\begin{equation}
\label{eq19}
\OPT(\tn{LP-MUHV}) \le \frac 12 k w_t + {k \choose 2} w_b.
\end{equation}

Combining Eqs.~(\ref{eq18}) and (\ref{eq19}) and setting $w_b = 0$, it gives a lower bound on the integrality gap of (\ref{lp:muhv}):
\[
\frac {\OPT(\tn{MUHV})}{\OPT(\tn{LP-MUHV})} \ge \frac {k-1}{\frac 12 k} = 2 - \frac 2k.
\]
We thus have proved the following theorem.

\begin{theorem}
\label{thm11}
The integrality gap of (\ref{lp:muhv}) has a lower bound of $2 - \frac 2k$.
\end{theorem}

Theorems \ref{thm9} and \ref{thm11} together imply that the $(2 - \frac 2k)$-approximation algorithm $\mcr\mcr$ for the MUHV problem is the best possible
based on the LP relaxation (\ref{lp:muhv}), and furthermore

\begin{corollary}
\label{coro12}
The $(2 - \frac 2k)$-approximation algorithm $\mcr\mcr$ for the Sub-ML problem is the best possible
based on the convex relaxation on the Lov\'{a}sz extension (\ref{cp:subml}).
\end{corollary}

In the {\em hypergraph multiway cut} (Hyp-MC) problem, 
we are given a hypergraph $H = (V_H, E_H)$ with a non-negative weight $w(e)$ for each hyperedge $e \in E_H$ and
a set of $k$ terminals $T = \{t_1, t_2, \ldots, t_k\} \subseteq V$.
The goal is to remove a minimum-weight set of hyperedges so that every two terminals are disconnected.
Ene {\it et al.}~\cite{EVW13} proved that a $(2 - \frac{2}{k} - \epsilon)$-approximation for Hyp-MC is NP-hard, for any $\epsilon > 0$,
assuming the Unique Games Conjecture.
We show next that it is also Unique Games-hard to achieve a $(2 - \frac{2}{k} - \epsilon)$-approximation for the MUHV problem.

\begin{theorem}
\label{thm13}
No $(2 - \frac{2}{k} - \epsilon)$-approximation algorithm for the MUHV problem exists, for any $\epsilon > 0$, assuming the Unique Games Conjecture.
\end{theorem}
\begin{proof}
We prove the theorem by constructing an approximation preserving reduction from the Hyp-MC problem to the MUHV problem.

Given an instance $(H = (V_H, E_H), w(\cdot), T = \{t_1, t_2, \ldots, t_k\})$ of the Hyp-MC problem, 
we construct an instance $(G = (V, E), w'(\cdot), C = \{1, 2, \ldots, k\}, c)$ of MUHV as follows:
\begin{itemize}
\item
	for each hyperedge $e \in E_H$, we create a vertex $v_e$;
	let the vertex set be $V = V_H \cup V_E$, where $V_E = \{v_e \mid e \in E_H\}$;
	call $T = \{t_1, t_2, \ldots, t_k\} \subseteq V$ the terminal set;
\item
	for each vertex $v \in V_H$, its weight is $w'(v) = 0$;
	for each vertex $v_e \in V_E$, its weight is $w'(v_e) = w(e)$;
\item
	for each vertex $v_e \in V_E$, it is adjacent to every vertex of $e$;
	let the edge set be $E = \{\{v_e, v\} \mid e \in E_H, v \in e\}$;
\item
	let the color set be $C = \{1, 2, \ldots, k\}$ and let the partial coloring function $c: V \mapsto C$ pre-color the terminal $t_i$ with $i$.
\end{itemize}
We note that the graph $G$ is actually bipartite, and the two parts of vertices are $V_H$ and $V_E$.

Consider a simple path $P$ connecting two terminals $t_i$ and $t_j$ in the hypergraph $H = (V_H, E_H)$.
Every two consecutive vertices on $P$ must belong to a common hyperedge;
therefore, the path $P$ one-to-one corresponds to a simple path in the constructed graph $G = (V, E)$ connecting the two vertices $t_i$ and $t_j$,
which is also denoted as $P$ without any ambiguity.
For any coloring function $c^*$ that completes the given partial coloring function $c$, we have $c^*(t_i) = i$ for each $i = \{1, 2, \ldots, k\}$.
It follows that any simple path $P$ connecting $t_i$ and $t_j$ must contain at least one vertex $v_e \in V_E$
such that its preceding vertex and its succeeding vertex, both in $V_H$, are colored differently.
The vertex $v_e$ is thus unhappy under the coloring scheme $c^*$.
In the hypergraph $H$, removing the corresponding hyperedge $e$ breaks the path $P$, thus disconnecting $t_i$ and $t_j$ via the path $P$.
Therefore, removing all the hyperedges whose corresponding vertices in the graph $G$ are unhappy disconnects all pairs of terminals.
In other words, any solution to the constructed instance of the MUHV problem can be transferred into a feasible solution to the given instance of the Hyp-MC problem;
the transfer is done in linear time and the two solutions have exactly the same value.

Conversely, given a subset $E^*_H$ of hyperedges in the hypergraph $H = (V_H, E_H)$ whose removal disconnects all pairs of terminals,
let $V^i_H$ and $E^i_H$ denote the subsets of vertices and hyperedges in the connected component of the remainder hypergraph $(V_H, E_H - E^*_H)$
that contains the terminal $t_i$, for each $i = 1, 2, \ldots, k$.
Denote the vertex subsets in the constructed graph $G = (V, E)$ corresponding to $V^i_H$ and $E^i_H$ as $V^i_H$ and $V^i_E$, respectively,
for $i = 1, 2, \ldots, k$.
We complete the partial coloring function $c$ by coloring all vertices of $V^i_H \cup V^i_E$ with the color $i$, for $i = 1, 2, \ldots, k$,
and coloring all the other remaining vertices of $V$ with the color $1$.
Clearly, all vertices of $\{v_e \mid e \in E_H - E^*_H\}$ are happy;
due to every vertex of $V_H$ has weight $0$ (such that we may ignore its happiness),
we conclude that the total weight of unhappy vertices in this coloring scheme is no more than $w(E^*_H) := \sum_{e \in E^*_H} w(e)$.

In summary, the Hyp-MC problem is polynomial-time reducible to the MUHV problem, and our reduction preserves the value of any feasible solution
and consequently preserves the approximation ratio.
\end{proof}

\begin{corollary}
\label{coro14}
No $(2 - \frac{2}{k} - \epsilon)$-approximation algorithm for the Sub-ML problem exists, for any $\epsilon > 0$, assuming the Unique Games Conjecture.
\end{corollary}

\section{Conclusions}
We studied the maximum happy vertices (MHV) problem and its complement, the minimum unhappy vertices (MUHV) problem.
We first showed that the MHV and MUHV problems are a special case of the supermodular and submodular multi-labeling (Sup-ML and Sub-ML) problems, respectively,
by re-writing the objective functions as set functions.
We next showed that the convex relaxation on the Lov\'{a}sz extension,
presented by Chekuri and Ene for the submodular multi-partitioning (Sub-MP) problem \cite{CE11smp},
can be extended for the Sub-ML problem,
thereby proving that the Sub-ML (Sup-ML, respectively) can be approximated within a factor of $2 - \frac{2}{k}$ ($\frac{2}{k}$, respectively).
These general results imply that the MHV and the MUHV problems can also be approximated within $\frac{2}{k}$ and $2 - \frac{2}{k}$, respectively,
using the same approximation algorithms.

For MHV, this $\frac{2}{k}$-approximation algorithm improves the previous best approximation ratio $\max \{\frac{1}{k}, \frac{1}{\Delta + 1}\}$~\cite{ZL15,ZJL15}, 
where $\Delta$ is the maximum vertex degree of the input graph.
We also showed that the LP relaxation presented by Zhang {\it et al.} \cite{ZJL15} is the same as
the concave relaxation on the Lov\'{a}sz extension for the Sup-ML problem;
we then proved an upper bound of $\frac{2}{k}$ on the integrality gap of the LP relaxation.
These suggest that the $\frac{2}{k}$-approximation algorithm is the best possible based on the LP relaxation;
thus the $\frac{2}{k}$-approximation algorithm is also the best possible based on the concave relaxation on the Lov\'{a}sz extension for the Sup-ML problem.

For MUHV, we formulated a novel LP relaxation and proved that it is the same as the convex relaxation on the Lov\'{a}sz extension for the Sub-ML problem;
we then showed a lower bound of $2 - \frac{2}{k}$ on the integrality gap of the LP relaxation.
Similarly, these suggest that the $(2 - \frac{2}{k})$-approximation algorithm is the best possible based on the LP relaxation;
thus the $(2 - \frac{2}{k})$-approximation algorithm is also the best possible based on the convex relaxation on the Lov\'{a}sz extension for the Sub-ML problem.
Lastly, we proved that this $(2 - \frac{2}{k})$-approximation is optimal for the MUHV problem, assuming the Unique Games Conjecture.
The last hardness result gives another evidence that
it is Unique Games-hard to achieve a $(2 - \frac{2}{k} - \epsilon)$-approximation for the general Sub-ML problem, for any $\epsilon > 0$.


\newpage
\appendix
\section{Properties of the boundary $\partial(\cdot)$}
\begin{lemma}
\label{lemma12}
Given a graph $G = (V, E)$, the boundary $\partial: 2^V \mapsto \mbr$ has the following properties:
i) $\partial(\emptyset) = \emptyset$;
ii) $\partial(X \cap Y) \subseteq \partial(X) \cup \partial(Y)$;
iii) $\partial(X \cup Y) \subseteq \partial(X) \cup \partial(Y)$; and
iv) $\partial(X \cap Y) \cap \partial(X \cup Y) \subseteq \partial(X) \cap \partial(Y)$, for any two subsets $X, Y \subseteq V$.
\end{lemma}
\begin{proof}
Recall that for any $X \subseteq V$, $\partial(X)$ is the subset of vertices of $X$ each has at least one neighbor outside of $X$.
It follows that $\partial(\emptyset) = \emptyset$.

Next, for any $v \in \partial(X \cap Y)$, $v \in X \cap Y$ and $v$ has a neighbor $u \notin X \cap Y$.
That is, $u$ is either outside of $X$ or outside of $Y$.
If $u$ is outside of $X$, then $v \in \partial(X)$;
otherwise, $v \in \partial(Y)$.
Therefore, $\partial(X \cap Y) \subseteq \partial(X) \cup \partial(Y)$.

For any $v \in \partial(X \cup Y)$, $v \in X \cup Y$ and $v$ has a neighbor $u \notin X \cup Y$.
If $v \in X$, then $v \in \partial(X)$;
otherwise, $v \in \partial(Y)$.
Therefore, $\partial(X \cup Y) \subseteq \partial(X) \cup \partial(Y)$.

Lastly, from the last paragraph, if $v \in \partial(X \cap Y) \cap \partial(X \cup Y)$,
then $v \in X \cap Y$ and $v$ has a neighbor $u \notin X \cup Y$.
These imply that $v \in \partial(X)$ and $v \in \partial(Y)$, i.e., $v \in \partial(X) \cap \partial(Y)$.
Therefore, $\partial(X \cap Y) \cap \partial(X \cup Y) \subseteq \partial(X) \cap \partial(Y)$.
\end{proof}

\section{Submodularity of the set function $f(\cdot)$ defined in Eq.~(\ref{eq1})}
\begin{proof}
Given a graph $G = (V, E)$, we want to prove that for any two subsets $X, Y \subseteq V$, $f(X) + f(Y) \ge f(X \cap Y) + f(X \cup Y)$,
where $f(X) := w(\partial(X))$.

Recall that the boundary $\partial: 2^V \mapsto \mbr$ satisfies
ii) $\partial(X \cap Y) \subseteq \partial(X) \cup \partial(Y)$ and iii) $\partial(X \cup Y) \subseteq \partial(X) \cup \partial(Y)$.
Therefore, $\partial(X \cap Y) \cup \partial(X \cup Y) \subseteq \partial(X) \cup \partial(Y)$ also holds.
Furthermore, the boundary $\partial: 2^V \mapsto \mbr$ also satisfies iv) $\partial(X \cap Y) \cap \partial(X \cup Y) \subseteq \partial(X) \cap \partial(Y)$.
We thus conclude that
\[
w(\partial(X \cap Y) \cup \partial(X \cup Y)) + w(\partial(X \cap Y) \cap \partial(X \cup Y)) \le w(\partial(X) \cup \partial(Y)) + w(\partial(X) \cap \partial(Y)),
\]
which is exactly
\[
f(X \cap Y) + f(X \cup Y)) \le f(X) + f(Y).
\]
This proves the submodularity (Lemma~\ref{lemma8}, and the supermodularity in Lemma~\ref{lemma3}).
\end{proof}

\section{An instance showing that Sub-ML does not reduce to Sub-MP}
The following instance of the MUHV problem shows that,
given a graph $G = (V, E)$ and some pre-colored vertices,
contracting the vertices pre-colored the same into a single vertex will change the objective function, 
resulting in an instance with a completely different optimum.

We set a constant $W > \epsilon > 0$.

In this instance $I$, the input graph $G = (V, E)$ has $9$ vertices,
each vertex $v_i$ has a non-negative weight $w(v_i)$,
the color set is $C = \{1, 2, 3\}$,
and the partial coloring function $c$ pre-colors $2$ vertices with each color.
In more details,
\[
\begin{array}{l}
V = \{v_1, v_2, \ldots, v_9\};\\
E = \{(v_i, v_{i+1}), (v_i, v_{i+2}), (v_{i+1}, v_{i+2}) | i = 1, 4, 7\} \cup \{(v_2, v_6), (v_5, v_9), (v_8, v_3)\};\\
w(v_1) = w(v_4) = w(v_7) = W > 0,\\
w(v_i) = \epsilon < W, \ i = 2, 3, 5, 6, 8, 9;\\
c(v_1) = c(v_2) = 1,\\
c(v_4) = c(v_5) = 2,\\
c(v_7) = c(v_8) = 3.
\end{array}
\]
Since $v_3, v_6, v_9$ must be unhappy and $W > \epsilon$, an optimal solution is to color $v_3$ with $1$, $v_6$ with $2$ and $v_9$ with $3$.
Then the minimum total weight of unhappy vertices is $6\epsilon$.
Observe that in this optimal solution, $v_1$ is happy but $v_2$ is not.

By contracting all the vertices pre-colored with the same color into a single vertex,
that is, contracting $v_1, v_2$ into $v_{12}$, contracting $v_4, v_5$ into $v_{45}$, and contracting $v_7, v_8$ into $v_{78}$,
we obtain an instance $I' = (G' = (V', E'), w', C, c')$ as follows:
\[
\begin{array}{l}
V' = \{v_{12}, v_3, v_{45}, v_6, v_{78}, v_9\};\\
E' = \{(v_{12}, v_3), (v_{45}, v_6), (v_{78}, v_9)\} \cup \{(v_{12}, v_6), (v_{45}, v_9), (v_{78}, v_3)\};\\
w'(v_{12}) = w'(v_{45}) = w'(v_{78}) = W + \epsilon,\\
w'(v_3) = w'(v_6) = w'(v_9) = \epsilon;\\
c'(v_{12}) = 1,\\
c'(v_{45}) = 2,\\
c'(v_{78}) = 3.
\end{array}
\]
Note that $G'$ is a simple circle.
Since $v_3, v_6, v_9$ must still be unhappy, at most one of $v_{12}$, $v_{45}$ and $v_{78}$ can become happy.
Thus the minimum total weight of unhappy vertices here is $2W + 5\epsilon > 6\epsilon$.
This optimal coloring scheme for $I'$ is certainly not optimal for $I$.

We remark that such a contracting procedure fails for one reason that the vertices pre-colored the same do not have the same neighbors.
For example, $v_1$ has a neighbor $v_3$ other than $v_2$ while $v_2$ has neighbors $v_3$ and $v_6$ other than $v_1$.
\end{document}